\documentclass[acmtoms]{acmtrans2m}

\acmVolume{xx}
\acmNumber{x}
\acmYear{xx}
\acmMonth{xx}

\usepackage{amsmath,amssymb}
\usepackage{graphicx}
\usepackage{xspace}
\usepackage{listings}
\usepackage{longtable}
\usepackage{enumerate}
\usepackage{verbatim}  
\newcommand{\bmat}[1]{{\renewcommand{\arraystretch}{1.1}%
   \begin{bmatrix}#1\end{bmatrix}}}

\ifx   \innerprod\undefined%
   \def\innerprod(#1,#2){\langle#1,#2\rangle} 
\fi

\newcommand{\Newcommand}[2]%
   {\ifx#1\undefined \newcommand{#1}{#2} \else \renewcommand{#1}{#2} \fi}
\ifx\mod\undefined
  \newcommand{\mod}[1]{|#1|}
\else
  \renewcommand{\mod}[1]{|#1|}
\fi
\Newcommand{\Re} {\mathbb{R}}           

\providecommand{\rank} {\mathop{\mathrm{rank}}}

\newcommand{\be}{\begin{enumerate}}
\newcommand{\ee}{\end{enumerate}}

\newcommand{\T}{^T\!}

\newcommand{\Cpp}{C\raise3pt\hbox{\tiny++}}

\newcommand{\cond}{\mathrm{cond}}

\newcommand{\Null}{\mathop{\mathrm{null}}}

\newcommand{\range}{\mathop{\mathrm{range}}}

\newcommand{\drop}{^{\null}}

\newcommand{\normm}[1]{\biggl\|#1\biggr\|}
\newcommand{\norm}[1]{\|#1\|}

\newcommand{\spose}[1]{\hbox to 0pt{#1\hss}}

\providecommand{\text}[1]{\hbox{\quad#1\quad}}

\newcommand{\nthinsp}{\mskip -2   mu}

\Newcommand{\R}{_{\scriptscriptstyle R}}

\newcommand{\superstar}{^{\raise 0.5pt\hbox{$\nthinsp *$}}}
\newcommand{\SUPERSTAR}{^{\raise 0.5pt\hbox{$*$}}}

\newcommand{\lamstarT }{\lambda^{\raise 0.5pt\hbox{$\nthinsp *$}T}}

\newcommand{\Abar}{\skew7\bar A}

\newcommand{\rbar}{\skew3\bar r}

\newcommand{\xbar}{\skew{2.8}\bar x}

\newcommand{\CG}       {{\small CG}\xspace} 
\newcommand{\CGLS}     {{\small CGLS}\xspace}

\newcommand{\LSQR}     {{\small LSQR}\xspace}
\newcommand{\GMRES}    {{\small GMRES}\xspace}
\newcommand{\LSMR}     {{\small LSMR}\xspace}
\newcommand{\MINRES}   {{\small MINRES}\xspace}
\newcommand{\MINRESQLP}{{\small MINRES-QLP}\xspace}
\newcommand{\SYMMLQ}   {{\small SYMMLQ}\xspace}
\newcommand{\SQMR}     {{\small SQMR}\xspace}
\newcommand{\MATLAB}   {{\small MATLAB}\xspace}
\newcommand{\FORTRAN}  {{\small FORTRAN}\xspace}
\newcommand{\SymOrtho} {\text{SymOrtho}\xspace}

\newtheorem{theorem}{Theorem}[section]

\newtheorem{lemma}[theorem]{Lemma}
\newdef{definition}[theorem]{Definition}
\newdef{remark}[theorem]{Remark}

\newcommand{\gamamin}{\underline{\gamma}}
\newcommand{\gamamax}{\overline{\gamma}}
\newcommand{\myhalf}{\frac12}

\newcommand{\underTk}{\underline{T_k}}
\newcommand{\underTkp}{\underline{T_{k+1}}}

\newcommand{\BibTeX}{{\rm B\kern-.05em{\sc i\kern-.025em b}
\kern-.08emT\kern-.1667em\lower.7ex\hbox{E}\kern-.125emX}}

\usepackage[vlined,ruled,nofillcomment,linesnumbered]{algorithm2e}
\SetAlgoSkip{}\SetAlgoInsideSkip{smallskip}

\newenvironment{algo}[1]
{\begin{algorithm}[#1]%
    \small%
    \DontPrintSemicolon%
    \SetArgSty{texttsf}%
    \SetTitleSty{textsf}{}%
    \SetNlSty{textrm}{}{}%
    \SetKwInput{Inputs}{input}%
    \SetKwInput{Outputs}{output}%
    \SetKwData{Converged}{converged}%
    \SetKwComment{tcc}{[}{]}
}
{\end{algorithm}}

{\ttfamily \begin{longtable}{#1}}%
{\end{longtable}}

\title{ALGORITHM xxx:
       MINRES-QLP for Singular Symmetric and Hermitian Linear
       Equations and Least-Squares Problems}
\author{SOU-CHENG T. CHOI
     \\ University of Chicago/Argonne National Laboratory
   \and MICHAEL A. SAUNDERS
     \\ Stanford University}

\begin{abstract}
  We describe algorithm MINRES-QLP and its FORTRAN~90 implementation
  for solving symmetric or Hermitian linear systems or least-squares
  problems. If the system is singular, MINRES-QLP computes the unique
  minimum-length solution (also known as the pseudoinverse solution),
  which generally eludes MINRES. In all cases, it overcomes a
  potential instability in the original MINRES algorithm.
  A positive-definite preconditioner may be supplied.  Our FORTRAN~90
  implementation illustrates a design pattern that allows users to
  make problem data known to the solver but hidden and secure from
  other program units.  In particular, we circumvent the need for
  reverse communication.
  While we focus here on a FORTRAN~90 implementation, we also provide
  and maintain MATLAB versions of MINRES and MINRES-QLP.
\end{abstract}






\begin{document}

\lstdefinestyle{numbers}
{numbers=left, stepnumber=1, numberblanklines=false,
numberstyle=\tiny, numbersep=8pt}
\lstdefinestyle{nonumbers}
{numbers=none}

\lstset{
language=[90]Fortran, 
basicstyle=\small,    
style=numbers,        
emptylines=*1,        
breaklines=true,      
escapeinside=<>,      
framesep=2.5mm,       
}

\begin{bottomstuff}
This work was supported in part by the Office of Advanced Scientific
Computing Research, Office of Science, U.S. Dept.~of Energy, under
Contract DE-AC02-06CH11357; National Science Foundation grant
CCR-0306662; Office of Naval Research grants N00014-02-1-0076 and
N00014-08-1-0191; and U.S. Army Research Laboratory through the Army
High Performance Computing Research Center.

Authors' addresses:
   S.-C. T. Choi, Computation Institute, University of Chicago,
   Chicago, IL 60637; email: sctchoi@uchicago.edu;
   M.~A. Saunders, Department of Management Science and Engineering,
   Stanford University, Stanford, CA 94305-4121; email:
   saunders@stanford.edu.
\end{bottomstuff}

\maketitle

\clearpage

\section{INTRODUCTION}

\MINRESQLP \cite{C06,CPS11} is a Krylov subspace method for computing
the minimum-length and minimum-residual solution (also known as the
pseudoinverse solution) $x$ to the following linear systems or
least-squares (LS) problems:
\begin{align}
     \text{solve }    & Ax = b,
     \label{eqn1a}
  \\ \text{minimize } & \norm{x}_2 \quad \text{s.t.} \quad
                        Ax = b,
     \label{eqn1b}
  \\ \text{minimize } & \norm{x}_2  \quad
                        \text{s.t.} \quad
                        x \in \arg\min_x \norm{Ax-b}_2,
      \label{eqn1c}
\end{align}
where $A$ is an $n \times n$ symmetric or Hermitian matrix and $b$ is
a real or complex $n$-vector.  Problems (\ref{eqn1a}) and
(\ref{eqn1b}) are treated as special cases of (\ref{eqn1c}).  The
matrix $A$ is usually large and sparse, and it may be singular.\footnote{A further input parameter $\sigma$ (a real shift
  parameter) causes \MINRESQLP to treat ``$A$'' as if it were
  $A-\sigma I$.  For example, ``singular $A$'' really means that
  $A-\sigma I$ is singular.}  It is defined by means of a user-written
subroutine \texttt{Aprod},
whose function is to compute the product $y=Av$ for any given vector
$v$.


Let $x_k$ be the solution estimate associated with \MINRESQLP's $k$th
iteration, with residual vector $r_k = b - Ax_k$.  Without loss of
generality, we define $x_0 = 0$.  \MINRESQLP provides recurrent
estimates of $\norm{x_k}$, $\norm{r_k}$, $\norm{Ar_k}$, $\norm{A}$,
$\cond(A)$, and $\norm{Ax_k}$, which are used in the stopping
conditions.

Other iterative methods specialized for symmetric systems $Ax=b$ are
the conjugate-gradient method (\CG) \cite{HS52}, \SYMMLQ and \MINRES
\cite{PS75}, and \SQMR \cite{FN94}.  Each method requires one product
$Av_k$ at each iteration for some vector $v_k$.  \CG is intended for
positive-definite $A$, whereas the other solvers allow $A$ to be
indefinite.

If $A$ is singular, \SYMMLQ requires the system to be consistent,
whereas \MINRES returns an LS solution for (\ref{eqn1c}) but generally
not the min-length solution; see \cite{C06,CPS11} for examples.  \SQMR
without preconditioning is mathematically equivalent to \MINRES but
could fail on a singular problem.  To date, \MINRESQLP is probably the
most suitable \CG-type method for solving (\ref{eqn1c}).

In some cases the more established symmetric methods may still be
preferable.

\begin{enumerate} 
\item If $A$ is positive definite, \CG minimizes the energy norm of
  the error \hbox{$\norm{x-x_k}_A$} in each Krylov subspace and
  requires slightly less work per iteration.  However, \CG, \MINRES,
  and \MINRESQLP do reduce $\norm{x-x_k}_A$ and $\norm{x-x_k}$
  monotonically.  Also, \MINRES and \MINRESQLP often reduce
  $\norm{r_k}$ to the desired level significantly sooner than does
  \CG, and the backward error for each $x_k$ decreases monotonically.
  (See Section~\ref{sect-norms-stopping-conds} and
  \cite{Fong2011,FS12}.)

\item If $A$ is indefinite but $Ax = b$ is consistent (e.g., if $A$ is
  nonsingular), \SYMMLQ requires slightly less work per iteration, and
  it reduces the error norm $\norm{x-x_k}$ monotonically.  \MINRES and
  \MINRESQLP \emph{usually} reduce $\norm{x-x_k}$
  \cite{Fong2011,FS12}.

\item If $A$ is indefinite and well-conditioned and $A x = b$ is
  consistent, \MINRES might be preferable to \MINRESQLP because it
  requires the same number of iterations but slightly less work per
  iteration.

\item \MINRES and \MINRESQLP require a preconditioner to be positive
  definite.  \SQMR might be preferred if $A$ is indefinite and an
  effective indefinite preconditioner is available.
\end{enumerate}


\MINRESQLP has two phases.  Iterations start in the \emph{\MINRES
  phase} and transfer to the \emph{\MINRESQLP phase} when a subproblem
(see \eqref{eqn:LSsubprob} below) becomes ill-conditioned by a certain
measure.  If every subproblem is of full rank and well-conditioned,
the problem can be solved entirely in the \MINRES phase, where the
cost per iteration is essentially the same as for \MINRES.  In the
\MINRESQLP phase, one more work vector and $5n$ more multiplications
are used per iteration.

\MINRESQLP described here is implemented in \FORTRAN~90 for real
double-precision problems.  It contains no machine-dependent constants
and does not need to use features such as polymorphism from
\FORTRAN~2003 or 2008. It requires an auxiliary subroutine
\texttt{Aprod} and, if a preconditioner is supplied, a second
subroutine \texttt{Msolve}. Since \FORTRAN~90 contains the intrinsic
\texttt{COMPLEX} data type, our implementation is also adapted for
complex problems.  Precision other than double can be handily obtained
by supplying different values to the data attribute \texttt{KIND}.
The program can be compiled with \FORTRAN~90 and \FORTRAN~95 compilers
such as \texttt{f90}, \texttt{f95}, \texttt{g95}, and
\texttt{gfortran}.  We also have a \MATLAB implementation, which is
capable of solving both real and complex problems readily. All
versions are available for download at \hbox{\citeA{SOL}}.

Table~\ref{table-notation} lists the main notation used.

\begin{table}[hb]   
\caption{Key notation.}
\label{table-notation}
\begin{center}
  \begin{tabular}{|l|l|}
    \hline
     $\norm{\cdot}$         & matrix or vector two-norm
   \\[.5ex] $\Abar$         & $\Abar = A - \sigma I$ (see also $\sigma$ below)
   \\ $\cond(A)$            & condition number of $A$ with respect
                             to two-norm = $\smash[b]{\frac{\max
                             |\lambda_i|}{\min_{\lambda_i \neq 0}
                             |\lambda_i|}}$
   \\ $e_i$                 & $i$th unit vector
   \\ $\ell$                & index of the last Lanczos iteration when $\beta_{\ell+1} = 0$
   \\ $n$                   & order of $A$
   \\ $\Null(A)$            & null space of $A$ defined as
                              $\{x \in \mathbb{R}^n \mid Ax = 0 \}$
   \\ $\range(A)$           & column space of $A$ defined as
                              $\{Ax \mid x \in \mathbb{R}^n\}$                            
   \\ $\T$                  & (right superscript to a vector or a matrix) transpose
   \\ $x^{\dagger}$          & unique minimum-length least-squares solution of problem~(\ref{eqn1c}) 
   \\ $\mathcal{K}_k(A,b)$  & $k$th Krylov subspace defined as
                              $\operatorname*{span}
                             \{b, Ab, \ldots, A^{k-1}b\}$                            
   \\ $\varepsilon$         & machine precision
   \\ $\sigma$              & scalar shift to diagonal of $A$
   \\[3pt] \hline
  \end{tabular}
\end{center}
\end{table}

\subsection{Least-Squares Methods}

Further existing methods that could be applied to \eqref{eqn1c} are
\CGLS and \LSQR \cite{PS82a,PS82b}, \LSMR \cite{FS11}, and \GMRES
\cite{SS86}, all of which reduce $\norm{r_k}$ monotonically.  The
first three methods would require two products $A v_k$ and $A u_k$
each iteration and would be generating points in less favorable
subspaces. \GMRES requires only products $A v_k$ and could use any
nonsingular (possibly indefinite) preconditioner.  It needs increasing
storage and work each iteration, perhaps requiring restarts, but it
could be more effective than \MINRES or \MINRESQLP (and the other
solvers) if few total iterations were required.
Table~\ref{table-algo-work} summarizes the computational requirements
of each method.

\begin{table}
\caption{Comparison of various least-squares solvers on $n \times n$
  systems (\ref{eqn1c}).  Storage refers to memory required by working
  vectors in the solvers. Work counts number of floating-point
  multiplications.  On inconsistent systems, all solvers below except
  MINRES and GMRES with restart parameter $m$ return the
  minimum-length LS solution (assuming no preconditioner).}
\label{table-algo-work}
\begin{center}
  \begin{tabular}{|l|c|c|c|c|}
    \hline
  \\[-10pt]
    Solver      & Storage  & Work per   & Products per  & Systems to Solve per 
  \\[-1pt]      &          & Iteration  &  Iteration    & Iteration with Preconditioner
  \\ \hline $ $ &          &   
  \\[-9.5pt]
     MINRES     & $7n$     & $9n$       & 1 &  1
  \\ MINRES-QLP &$7n$--$8n$& $9n$--$14n$& 1 &  1
  \\ GMRES($m$) & $(m+2)n$ & $(m+3+1/m)n$ & 1 &  1
  \\ CGLS       & $4n$     & $5n$       & 2 &  2
  \\ LSQR       & $5n$     & $8n$       & 2 &  2
  \\ LSMR       & $6n$     & $9n$       & 2 &  2
  \\ \hline
  \end{tabular}
\end{center}
\end{table}

\subsection{Regularization}    
\label{sect-regularization}

We do not discourage using \CGLS, \LSQR, or \LSMR if the goal is to
regularize an ill-posed problem using a small damping factor $\lambda
> 0$ as follows:
\begin{equation} \min_x \,\normm{\bmat{A \\ \lambda I}x-\bmat{b\\0}}.
       \label{eqn-reg}
\end{equation}
However, this approach destroys the original problem's symmetry.  The
normal equation of (\ref{eqn-reg}) is ($A^2 + \lambda^2 I) x = Ab$,
which suggests that a diagonal shift to $A$ may well serve the same
purpose in some cases. For symmetric positive-definite $A$, $\Abar =
A-\sigma I$ with $\sigma < 0$ enjoys a smaller condition number.  When
$A$ is indefinite, a good choice of $\sigma$ may not exist, for
example, if the eigenvalues of $A$ were symmetrically positioned
around zero. When this symmetric form is applicable, it is convenient
in \MINRES and \MINRESQLP; see (\ref{eqn1c}), (\ref{recur}), and
(\ref{precon-lanzos}).  We also remark that \MINRES and \MINRESQLP
produce good estimates of the largest and smallest singular values of
$\Abar$ (via diagonal values of $R_k$ or $L_k$ in (\ref{right_ortho})
and (\ref{Lsubproblem}); see \cite[Section~4]{CPS11}).

Three other regularization tools in the literature (see
\cite[Sections~12.1.1-12.1.3]{GV} and \cite{H98}) are LSQI,
cross-validation, and L-curve. LSQI involves solving a nonlinear
equation and is not immediately compatible with the Lanczos
framework. Cross-validation takes one row out at a time and thus does
not preserve symmetry.  The L-curve approach for a CG-type method
takes iteration $k$ as the regularization parameter \cite[Chapter
  8]{H98} if both $\norm{r_k}$ and $\norm{x_k}$ are monotonic.  By
design, $\norm{r_k}$ is monotonic in \MINRES and \MINRESQLP, and so is
$\norm{x_k}$ when $\Abar$ is positive definite \cite{Fong2011}.
Otherwise, we prefer the condition L-curve approach in \cite{CLR00},
which graphs $\cond(T_k)$ against $\norm{r_k}$. Yet another L-curve
feasible in \MINRESQLP is $\norm{x_{k-2}^{(2)}}$ against $\norm{r_k}$,
since the former is also monotonic (but available two iterations in
lag); see Section~\ref{sect-norms-stopping-conds}.



\section{MATHEMATICAL BACKGROUND}

Notation and details of algorithmic development from \cite{C06,CPS11}
are summarized here.

\subsection{Lanczos Process}

\MINRES and \MINRESQLP use the symmetric Lanczos process \cite{L50} to
reduce $A$ to a tridiagonal form $\underline{T_k}$. The process is
initialised with $v_0 \equiv 0$, $\beta_1 = \norm{b}$, and $\beta_1
v_1=b$.  After $k$ steps of the tridiagonalization, we have produced
\begin{align}
      p_k &= Av_k - \sigma v_k, \qquad \alpha_k = v_k\T p_k,
     \qquad \beta_{k+1}v_{k+1} = p_k-\alpha_kv_k-\beta_kv_{k-1},
   \label{recur}
\end{align}
where we choose $\beta_k > 0$ to give $\norm{v_k}=1$.
Numerically,
\begin{align*}
p_k = Av_k - \sigma v_k -\beta_kv_{k-1},
\qquad \alpha_k = v_k\T p_k, \qquad \beta_{k+1}v_{k+1}
= p_k - \alpha_k v_k
\end{align*}
is slightly better than (\ref{recur}) \cite{P76}, but we can express
(\ref{recur}) in matrix form:
\begin{align}
  V_k &\equiv  \bmat{v_1 & \!\cdots\! & v_k},
  \qquad AV_k = V_{k+1}\underline{T_k},
  \qquad \underline{T_k}  \equiv  \bmat{T_k \\ \beta_{k+1}e_k\T},
  \label{recur_matrix}
\end{align}
where $T_k = \textrm{tridiag}(\beta_i, \alpha_i, \beta_{i+1})$, $i=1,
\dots, k$.
In exact arithmetic, the Lanczos vectors in the columns of $V_k$ are
orthonormal, and the process stops with $k = \ell$ when
$\beta_{\ell+1}=0$ for some $\ell \le n$, and then $AV_\ell = V_\ell
T_\ell$. The rank of $T_\ell$ could be $\ell$ or $\ell-1$ (see
Theorem~\ref{theorem-Tk}).

\subsection{MINRES Phase}

\MINRESQLP typically starts with a \MINRES phase, which applies a
series of reflectors $Q_k$ to transform $\underline{T_k}$ to an upper
triangular matrix $\underline{R_k}$:
\begin{align}
       Q_k \bmat{\,\underline{T_k} & \beta_1 e_1} &= \bmat{R_k & t_k
         \\ 0 & \phi_k} \equiv\bmat{\,\underline{R_k} &
         \bar{t}_{k+1}}, \label{right_ortho}
\end{align}
where
\begin{align*}
      Q_k &= Q_{k,k+1} \bmat{Q_{k-1} \\ & 1},   \qquad
      Q_{k,k+1}  \equiv
       \left[\begin{smallmatrix}
           I_{k-1}&&
        \\        & c_k &   \!\!\phantom-s_k
        \\        & s_k &   \!-c_k
        \end{smallmatrix}\right]\!\!.      \nonumber \end{align*}
In the $k$th step, $Q_{k,k+1}$ is effectively a Householder reflector
of dimension 2 \cite[Exercise 10.4]{TB}; and its action including its
effect on later columns of $T_j$, $k < j \le \ell$, is compactly
described by
\begin{equation*}  \label{min7}
  \bmat{   c_k & \!\!\!\phantom-s_k
        \\ s_k & \!\! -c_k}
  \bmat{\begin{matrix}
           \gamma_k & \delta_{k+1} & 0
        \\ \beta_{k+1}    & \alpha_{k+1}       & \beta_{k+2}
        \end{matrix}
        & \biggm| &
        \begin{matrix} \phi_{k-1} \\ 0 \end{matrix}
       }
=
  \bmat{\begin{matrix}
           \gamma_k^{(2)} & \delta_{k+1}^{(2)} & \epsilon_{k+2}
        \\ 0              & \gamma_{k+1} & \delta_{k+2}
        \end{matrix}
        & \biggm| &
        \begin{matrix} \tau_k \\ \phi_k \end{matrix}
       },
\end{equation*}
where the superscripts with numbers in parentheses indicate the number
of times the values have been modified.  The $k$th solution
approximation to (\ref{eqn1c}) is then defined to be $x_k = V_k y_k$,
where $y_k$ solves the subproblem
\begin{equation}
  y_k = \arg\min_{ y \in \mathbb{R}^k }
        \norm{ \underline{T_k} y - \beta_1 e_1 }
        = \arg\min_{ y \in \mathbb{R}^k }
        \norm{ \underline{R_k} y - {\bar t}_{k+1} }.
  \label{eqn:LSsubprob}
\end{equation}
When $k < \ell$, $R_k$ is nonsingular and the unique solution of the
above subproblem satisfies $R_k y_k = t_k$.  Instead of solving for
$y_k$, \MINRES solves $R_k\T D_k\T = V_k\T$ by forward substitution,
obtaining the last column $d_k$ of $D_k$ at iteration $k$.  At the
same time, it updates $x_k \in \mathcal{K}_k(A,b)$ (see
Table~\ref{table-notation} for definition) via $x_0 \equiv 0$ and
\begin{equation}
    x_k = V_k y_k = D_k R_k y_k = D_k t_k
  = x_{k-1} + \tau_k d_k,\quad \tau_k\equiv e_k^Tt_k,
\label{minresxk}
\end{equation}
where one can show using $V_k = D_k R_k$ that
$  d_{k} = ({v_{k}-\delta_{k}^{(2)}d_{k-1} - \epsilon_{k}
   d_{k-2}}) / {\gamma_{k}^{(2)}}. 
$

\subsection{MINRES-QLP Phase}

The \MINRES phase transfers to the \MINRESQLP phase when an estimate
of the condition number of $A$ exceeds an input parameter
$\mathit{trancond}$.  Thus, $\mathit{trancond} > 1/\varepsilon$ leads
to \MINRES iterates throughout (where $\varepsilon \approx 10^{-16}$
denotes the floating-point precision), whereas $\mathit{trancond} = 1$
generates \MINRESQLP iterates from the start.

Suppose for now that there is no \MINRES phase. Then \MINRESQLP
applies left reflections as in (\ref{right_ortho}) and a further
series of right reflections to transform $R_k$ to a lower triangular
matrix $L_k = R_k P_k$, where
\begin{align*}
     P_k &=  P_{1,2} \ \ P_{1,3} P_{2,3} 
       \ \cdots \ \ P_{k-2,k} P_{k-1,k},
\\ P_{k-2,k} &=  \left[\begin{smallmatrix}
           I_{k-3}&&
        \\        & c_{k2} &    & \!\!\phantom-s_{k2}
        \\        &        & 1  &
        \\        & s_{k2} &    & \!-c_{k2}
        \end{smallmatrix}\right]\!\!,
\qquad P_{k-1,k} =  \left[\begin{smallmatrix}
           I_{k-2}&&
        \\        & c_{k3} &   \!\!\phantom-s_{k3}
        \\        & s_{k3} &   \!-c_{k3}
        \end{smallmatrix}\right]\!\!.
\end{align*}
In the $k$th step, the actions of $P_{k-2,k}$ and $P_{k-1,k}$ are
compactly described by
\begin{align}
 &\hspace*{13pt}
  \bmat{\gamma_{k-2}^{(5)} &                    & \epsilon_k
     \\ \vartheta_{k-1}    & \gamma_{k-1}^{(4)} & \delta_k^{(2)}
     \\                    &                    & \gamma_{k}^{(2)}
       }
  \bmat{c_{k2} &   & \!\!\!\phantom-s_{k2}
     \\        & 1 &
     \\ s_{k2} &   & \!\!-c_{k2}
       }
  \bmat{ 1
     \\    & c_{k3} & \!\!\!\phantom-s_{k3}
     \\    & s_{k3} & \!\!-c_{k3}
       }
       \nonumber
\\ &=
  \bmat{\gamma_{k-2}^{(6)}\vspace{3pt} 
     \\ \vartheta_{k-1}^{(2)} & \gamma_{k-1}^{(4)}  & \delta_k^{(3)}
     \\ \eta_{k}              &                     & \gamma_{k}^{(3)}
       }
  \bmat{ 1 \vspace{3pt} 
     \\    & c_{k3} & \!\!\!\phantom-s_{k3} \vspace{3pt} 
     \\    & s_{k3} & \!\!-c_{k3}
       }
= \bmat{\gamma_{k-2}^{(6)}\vspace{3pt} 
     \\ \vartheta_{k-1}^{(2)} & \gamma_{k-1}^{(5)}  &
     \\ \eta_{k}              & \vartheta_{k}       & \gamma_{k}^{(4)}
       }.                        \label{Lk-compact}
\end{align}
The $k$th approximate solution to (\ref{eqn1c}) is then defined to be
$x_k = V_k y_k = V_k P_k u_k = W_k u_k$, where $u_k$ solves the
subproblem
\begin{equation}  \label{Lsubproblem}
  u_k\equiv \arg \min_u \norm{u} \quad\text{s.t.}\quad
        u \in \arg \min_{u \in \mathbb{R}^k}
          \, \left\|\bmat{L_k \\ 0} u - \bmat{t_k \\ \phi_k}\right\|.
\end{equation}
For $k<\ell$, $R_k$ and $L_k$ are nonsingular because
$\underline{T_k}$ has full column rank by
Lemma~\ref{lemma-underline-Tk} below.  It is only when $k=\ell$ and
$b\notin \range(A)$ that $R_k$ and $L_k$ are singular with rank
$\ell-1$ by Theorem~\ref{theorem-Tk}, in which case one can show that
$\eta_{k} = \gamma_{k}^{(3)} = \vartheta_{k} = \gamma_{k}^{(4)} = 0$
in~(\ref{Lk-compact}) and $L_\ell = \left[\begin{smallmatrix}
    L_{\ell-1} & 0 \\ 0 & 0 \end{smallmatrix}\right]$ with
$L_{\ell-1}$ nonsingular.  In any case, we need to solve only the
nonsingular lower triangular systems $L_k u_k = t_k$ or $L_{\ell-1}
u_{\ell-1} = t_{\ell-1}$. Then, $u_k$ and $y_k = P_k u_k$ are the
min-length solutions of (\ref{Lsubproblem}) and (\ref{eqn:LSsubprob}),
respectively.

\MINRESQLP updates $x_{k-2}$ to obtain $x_k$ by short-recurrence
orthogonal steps:
\begin{align}
   x_{k-2}^{(2)} &= x_{k-3}^{(2)} + \mu_{k-2}^{(3)} w_{k-2}^{(4)}
   \text{, where } x_{k-3}^{(2)} \equiv W_{k-3}^{(4)} u_{k-3}^{(3)},
   \label{qlpeqnsol1} \\ x_k &= x_{k-2}^{(2)} +
   \mu_{k-1}^{(2)} w_{k-1}^{(3)} + \mu_k w_k^{(2)}.
   \label{qlpeqnsol2}
\end{align}
Here $w_j$ refers to the $j$th column of $W_k = V_k P_k$, and $\mu_i$
is the $i$th element of $u_k$.

If this phase is preceded by a \MINRES phase of $k$ iterations ($0 < k
< \ell$), it starts by transferring the last three vectors $d_{k-2}$,
$d_{k-1}$, $d_k$ to $w_{k-2}$, $w_{k-1}$, $w_k$, and the solution
estimate $x_k$ from (\ref{minresxk}) to $x_{k-2}^{(2)}$ in
(\ref{qlpeqnsol1}).  This needs the last two rows of $L_k u_k = t_k$
(to give $\mu_{k-1}$, $\mu_k$) and the relations $W_k = D_k L_k$ and
$x_{k-2}^{(2)} = x_k - \mu_{k-1} w_{k-1} - \mu_k w_k$.
The cheaply available right reflections $P_k$ and the bottom right $3
\times 3$ submatrix of $L_k$ (i.e., the last term
in~(\ref{Lk-compact})) need to have been saved in the \MINRES phase in
order to facilitate the transfer.

\subsection{Norm Estimates and Stopping Conditions}
\label{sect-norms-stopping-conds}

Short-term recurrences are used to estimate the following quantities:
\begin{align*}
  \norm{r_k} &\approx \phi_k = \phi_{k-1} s_k,
             &\phi_0 &= \norm{b} 
             &(\phi_k \searrow)
\\[1ex]
  \norm{Ar_k}&\approx \psi_k = \phi_k \norm{[\gamma_{k+1}\ \;\delta_{k+2}]},
  &  &
  & (\psi_\ell = 0)
\\[1ex]
   \norm{x_k^{(2)}} &\approx \chi_{k-2}^{(2)} =  \norm{[ \chi^{(2)}_{k-3}  \  \; \mu_{k-2}^{(3)}]},
    &  \chi_{-2}\drop  = \chi_{-1}\drop   &= 0 
    & (\chi_{k-2}^{(2)} \nearrow)
\\[1ex]
  \norm{x_k} &\approx \chi_k\drop = \norm{[\chi_{k-2}^{(2)}\ \;\mu_{k-1}^{(2)}\ \;\mu_k]},
             &\chi_0\drop &=  0
             &(\chi_{\ell}\drop = \norm{x^\dagger})
\\[1ex]
  \norm{A x_k} &\approx \omega_k = \norm{[\omega_{k-1} \ \; \tau_k]},
           &\omega_0 &= 0
           &(\omega_k \nearrow)
\\[1ex]
  \norm{A} &\approx \mathcal{A}_k =
            \max \left\{ \mathcal{A}_{k-1}, \norm{\underTk e_k}, \gamamax_k \right\},
           &\mathcal{A}_0 &= 0
           &(\mathcal{A}_k \nearrow  \norm{A})
\\[0.5ex]
  \cond(A) &\approx \kappa_k = \mathcal{A}_k / \gamamin_k,
           &\kappa_0 &= 1
           &(\kappa_k \nearrow \cond(A))
\end{align*}
where $\gamamax_k$ and $\gamamin_k$ are the largest and smallest
absolute values of diagonals of $L_k$, respectively.  The up (down)
arrows in parentheses indicate that the quantities are monotonic
increasing (decreasing) if such properties exist. The last two
estimates tend to their targets from below; see~\cite{C06,CPS11} for
derivation.

\MINRESQLP has 14 possible stopping conditions in five classes that
use the above estimates and optional user-input parameters
$\mathit{itnlim}$, $\mathit{rtol}$, $\mathit{Acondlim}$, and
$\mathit{maxxnorm}$:
\begin{enumerate}[\upshape (C1)]
\item From Lanczos and the QLP factorization:  \label{C1}
\begin{align*}
  k = \mathit{itnlim};
  \qquad
  \beta_{k+1} < \varepsilon;
  \qquad
  \bigl| \gamma_k^{(4)} \bigr| < \varepsilon;
\end{align*}
\item Normwise relative backward errors (NRBE) \cite{PS02}: \label{C2}
\begin{align*}
  \norm{r_k}  / \left( \norm{A} \norm{x_k} + \norm{b} \right)
              \le \max(\mathit{rtol},\varepsilon);
  \qquad
  \norm{A r_k} / \left( \norm{A} \norm{r_k} \right)
              \le \max(\mathit{rtol},\varepsilon);
\end{align*}
\item Regularization attempts: \label{C3}
\begin{align*}
  \cond(A) \ge \mathit{\min(Acondlim,0.1/\varepsilon)};
  \qquad
  \norm{x_k}  \ge \mathit{maxxnorm};
\end{align*}
\item Degenerate cases: \label{C4}
\begin{align*}
   \beta_1 &=0  \qquad \Rightarrow \qquad b\phantom{_2}=0 \qquad  \Rightarrow
                \qquad x=0 \mbox{ is the solution};
\\ \beta_2 &=0  \qquad \Rightarrow \qquad v_2=0
                \qquad \Rightarrow \qquad A b = \alpha_1 b,
\\         &    \qquad \mbox{i.e., } b \mbox{ and } \alpha_1
                \mbox{ are an eigenpair of } A,
                \mbox{ and } x=b/\alpha_1\mbox{ solves } A x=b;
\end{align*}
\item Erroneous inputs: \label{C5}
\begin{align*}
    A \mbox{ not symmetric}; \qquad
    M \mbox{ not symmetric}; \qquad 
    M \mbox{ not positive definite};
\end{align*}
where $M$ is a preconditioner to be described in the next section. For
symmetry of $A$, it is not practical to check $e_i\T A e_j = e_j\T A
e_i$ for all $i,j=1,\dots,n$. Instead, we statistically test whether
$z=|x\T(Ay)-y\T(Ax)|$ is sufficiently small for two nonzero
$n$-vectors $x$ and $y$ (e.g., each element in the vectors is drawn
from the standard normal distribution).  For positive definiteness of
$M$, since $M$ is positive definite if and only if $M^{-1}$ is
positive definite, we simply test that $z_k\T M^{-1} z_k = z_k\T q_k >
0$ each iteration (see Section~\ref{sect-precond}).
\end{enumerate}

We find that the recurrence relations for $\phi_k$ and $\psi_k$ hold
to high accuracy. Thus $x_k$ is an acceptable solution of
(\ref{eqn1c}) if the computed value of $\phi_k $ or $\psi_k $ is
suitably small according to the NRBE tests in class~(C\ref{C2})
above. When a condition in~(C\ref{C3}) is met, the final $x_k$ may or
may not be an acceptable solution.

The class~(C\ref{C1}) tests for small $\beta_{k+1}$ and
$\gamma_k^{(4)}$ are included in the unlikely case in practice that
the theoretical Lanczos termination occurs.  Ideally one of the NRBE
tests should cause \MINRESQLP to terminate.  If not, it is an
indication that the problem is very ill-conditioned, in which case the
regularization and preconditioning techniques of
Sections~\ref{sect-regularization} and~\ref{sect-precond} may be
helpful.


\subsection{Two Theorems}

We complete this section by presenting two theorems from \cite{CPS11}
with slightly simpler proofs.

\begin{lemma}\label{lemma-underline-Tk}
$\rank(\underline{T_k}) =  k$ for all $k<\ell$.
\end{lemma}

\begin{proof}
  For $k < \ell$ we have $\beta_1,\dots,\beta_{k+1} > 0$ by
  definition.  Hence $\underline{T_k}$ has full column rank.
\end{proof}

\begin{theorem} \label{theorem-Tk}
$T_\ell$ is nonsingular if and only if $b \in \range(A)$.
  Furthermore, $\rank(T_\ell) = \ell-1$ if $b \notin \range(A)$.
\end{theorem}

\begin{proof}
We use $A V_\ell = V_\ell T_\ell$ twice.  First, if $T_\ell$ is
nonsingular, we can solve $T_\ell y_\ell = \beta_1 e_1$ and then $A
V_\ell y_\ell = V_\ell T_\ell y_\ell = V_\ell \beta_1 e_1 = b$.
Conversely, if $b \in \range(A)$, then $\range(V_\ell) \subseteq
\range(A)$.  Suppose $T_\ell$ is singular. Then there exists $z \ne 0$
such that $V_\ell T_\ell z = A V_\ell z=0$. That is, $0 \ne V_\ell z
\in \Null(A)$.  But this is impossible because $V_\ell z \in
\range(A)$ and $\Null(A) \cap \range(V_\ell) = 0$.  Thus, $T_\ell$
must be nonsingular.

We have shown that if $b \notin \range(A)$, $T_\ell =
\bmat{\underline{T_{\ell-1}}&
      \begin{smallmatrix}\beta_\ell e_{\ell-1}
                         \\[1pt] \alpha_\ell
      \end{smallmatrix}}$
is singular, and therefore $\ell > \rank(T_\ell) \ge
\rank(\underline{T_{\ell-1}}) = \ell-1$ by
Lemma~\ref{lemma-underline-Tk}. Therefore, $\rank(T_\ell) = \ell-1$.
\end{proof}

By Lemma~\ref{lemma-underline-Tk} and Theorem~\ref{theorem-Tk} we are
assured that the QLP decomposition without column pivoting
\cite{S99,CPS11} for $\underline{T_k}$ is rank-revealing, which is a
necessary precondition for solving a least-squares problem.

\begin{theorem}   \label{theorem-MINRES-QLP}
In \MINRESQLP, $x_\ell$ is the minimum-length solution of
(\ref{eqn1c}).
\end{theorem}

\begin{proof}

$y_\ell$ comes from the min-length LS solution of $T_\ell y_\ell
  \approx \beta_1 e_1$ and thus satisfies the normal equation $T_\ell
  ^2 y_\ell = T_\ell \beta_1 e_1$ and $y_\ell \in \range(T_\ell)$.
  Now $x_\ell = V_\ell y_\ell$ and $A x_\ell = A V_\ell y_\ell =
  V_\ell T_\ell y_\ell$.  Hence $A^2 x_\ell = A V_\ell T_\ell y_\ell =
  V_\ell T_\ell ^2 y_\ell = V_\ell T_\ell \beta_1 e_1 = A b$. Thus
  $x_\ell$ is an LS solution of (\ref{eqn1c}).  Since $y_\ell \in
  \range(T_\ell)$, $y_\ell = T_\ell z$ for some $z$, and so $x_\ell =
  V_\ell y_\ell = V_\ell T_\ell z = A V_\ell z \in \range(A)$ is the
  min-length LS solution of (\ref{eqn1c}).
\end{proof}

\section{PRECONDITIONING} \label{sect-precond}

Iterative methods can be accelerated if preconditioners are available
and well-chosen. For \MINRESQLP, we want to choose a symmetric
positive-definite matrix $M$ to solve a nonsingular system
(\ref{eqn1a}) by implicitly solving an equivalent symmetric consistent
system $M^{-\myhalf}A M^{-\myhalf}\xbar=\bar{b}$, where $M^{\myhalf}
x=\xbar$, $\bar{b} = M^{-\myhalf}b$, and $\cond{(M^{-\myhalf}A
  M^{-\myhalf})} \ll \cond{(A)}$.  This two-sided preconditioning
preserves symmetry. Thus we can derive preconditioned \MINRESQLP by
applying \MINRESQLP to the equivalent problem and obtain
$x=M^{-\myhalf}\xbar$.

With preconditioned \MINRESQLP, we can solve a singular consistent
system~(\ref{eqn1b}), but we will obtain a least-squares solution that
is not necessarily the minimum-length solution (unless $M=I$).  For
inconsistent systems (\ref{eqn1c}), preconditioning alters the
least-squares norm to $\norm{\cdot}_{M^{-1}}$, and the solution is of
minimum length in the new norm space. We refer readers to
\cite[Section 7]{CPS11} for a detailed discussion of various
approaches to preserving the two-norm ``minimum length.''

To derive \MINRESQLP, we define
\begin{equation}  \label{pminresd4}
   z_k = \beta_k M^{ \myhalf} v_k, \qquad
   q_k = \beta_k M^{-\myhalf} v_k,
   \qquad \text{so that} \quad M q_k =z_k.
\end{equation}
Then \( \beta_k = \norm{\beta_k v_k} = \norm{M^{-\myhalf} \! z_k} =
\norm{ z_k} _{M^{-1}} = \norm{ q_k }_M = \sqrt{ q_k\T z_k }, \) where
the square root is well defined because $M$ is positive definite, and
the following expressions replace the quantities in (\ref{recur}) in
the Lanczos iterations:
\begin{equation}
\label{precon-lanzos}
        p_k = A q_k - \sigma q_k, \qquad \alpha_k =
        \frac{1}{\beta_k^2} q_k\T p_k, \qquad z_{k+1} =
        \frac{1}{\beta_k} p_k - \frac{\alpha_k}{\beta_k} z_k -
        \frac{\beta_k}{\beta_{k-1}} z_{k-1}.
\end{equation}
We also need to solve the system $M q_k = z_k$ in (\ref{pminresd4}) at
each iteration.

In the \MINRES phase, we define $\bar{d}_k = M^{-\myhalf}d_k$ and
update the solution of the original problem (\ref{eqn1a}) by
\[
  \bar{d}_k = \Bigl( \frac{1}{\beta_k} q_k
                - \delta_k^{(2)} \bar{d}_{k-1}
                - \epsilon_k \bar{d}_{k-2} \Bigr)
                 \!\bigm/\! \gamma_k^{(2)},
  \qquad
  x_k = M^{-\myhalf} \xbar_k
      = x_{k-1} + \tau_k \bar{d}_k.       
\]

In the \MINRESQLP phase, we define $\overline{W}_k \equiv M^{-\myhalf}
W_k = (M^{-\myhalf} V_k) P_k $ and update the solution estimate of
problem (\ref{eqn1a}) by orthogonal steps:
\begin{align*}
   \bar{w}_k &= -(c_{k2} / \beta_k) q_k + s_{k2} \bar{w}_{k-2}^{(3)},
   \qquad \bar{w}_{k-2}^{(4)} = (s_{k2} / \beta_k) q_k
                               + c_{k2} \bar{w}_{k-2}^{(3)},
\\   \bar{w}_k^{(2)} &= s_{k3} \bar{w}_{k-1}^{(2)} - c_{k3} \bar{w}_k,
   \qquad \quad \quad \quad \bar{w}_{k-1}^{(3)}
                      = c_{k3} \bar{w}_{k-1}^{(2)} + s_{k3} \bar{w}_k,
\\  x_{k-2}^{(2)} &= x_{k-3}^{(2)} + \mu_{k-2}^{(3)} \bar{w}_{k-2}^{(4)},
   \qquad \qquad \quad  x_k = x_{k-2}^{(2)}
                            + \mu_{k-1}^{(2)} \bar{w}_{k-1}^{(3)}
                            + \mu_k \bar{w}_k^{(2)}.
\end{align*}
Let $\rbar_k = \bar{b} - M^{-\myhalf}A M^{-\myhalf} \xbar_k =
M^{-\myhalf} r_k$.  Then $x_k = M^{-\myhalf} \xbar_k$ is an acceptable
solution of (\ref{eqn1a}) if the computed value of $\phi_k \approx
\norm{\rbar_k} = \norm{r_k}_{M^{-1}}$ is sufficiently small.

We can now present our pseudocode in Algorithm~\ref{algo-pminresqlp}.
The reflectors are implemented in Algorithm~\ref{algo-symortho}
$\SymOrtho(a,b)$ for real $a$ and $b$, which is a stable form for
computing $r = \sqrt{a^2+b^2} \ge 0$ , $c = \frac{a}{r}$, and $s =
\frac{b}{r}$. The complexity is at most 6 flops and a square
root. Algorithm~\ref{algo-pminresqlp} lists all steps of \MINRESQLP
with preconditioning. For simplicity, $\bar{w}_k$ is written as $w_k$
for all relevant $k$. Also, the output $x$ solves $Ax \approx b$, but
other outputs are associated with the preconditioned system.

\begin{algo}{htp}

  \Inputs{$A, b, \sigma, M$}

  \smallskip

  $z_0 = 0$,
  \quad $z_1 = b$,
  \quad Solve $Mq_1 = z_1$,
  \quad $\beta_1 = \sqrt{b\T q_1}$,
  \quad $\phi_0 \!=\! \beta_1$
  \tcc*[r]{Initialize}

  $w_0 = w_{-1} = 0$,
  \qquad $x_{-2} = x_{-1} = x_0 = 0$\;
  $c_{0,1} \!=\! c_{0,2} \!=\! c_{0,3} \!=\! -1$,
  \quad $s_{0,1} \!=\! s_{0,2} \!=\! s_{0,3}\!=\! 0$,
  \quad $\tau_0 \!=\! \omega_0  \!=\! \chi_{-2}\drop
                \!=\! \chi_{-1}\drop \!=\! \chi_0\drop \!=\! 0$\;

  $\kappa_0 = 1$, \quad  $\mathcal{A}_0 = \delta_1 = \gamma_{-1} = \gamma_0
   = \eta_{-1} = \eta_0 = \eta_1
   = \vartheta_{-1} = \vartheta_0 = \vartheta_1
   = \mu_{-1} = \mu_0 = 0$
   
  $k=0$

  \BlankLine

  \While{no stopping condition is satisfied}{
    $k\leftarrow k+1$\;

    $p_k = Aq_k - \sigma q_k$,
    \qquad $\alpha_k = \frac{1}{\beta_k^2}q_k\T p_k$
    \tcc*[r]{Preconditioned Lanczos}

    $z_{k+1} = \frac{1}{\beta_k} p_k - \frac{\alpha_k}{\beta_k}z_k
             - \frac{\beta_k}{\beta_{k-1}} z_{k-1}$

    Solve $Mq_{k+1} = z_{k+1}$,
    \qquad $\beta_{k+1} = \sqrt{q_{k+1}\T z_{k+1}}$

    \lIf{$k = 1$}{$\rho_k = \norm{[\alpha_k \ \; \beta_{k+1}]}$}
    \lElse{$\rho_k = \norm{[\beta_k \ \; \alpha_k \ \; \beta_{k+1}]}$}

    $\delta_k^{(2)} = c_{k-1,1} \delta_k+s_{k-1,1} \alpha_k$
    \tcc*[r]{Previous left reflection\dots}

    $\gamma_k = s_{k-1,1}\delta_k -c_{k-1,1} \alpha_k$
    \tcc*[r]{on middle two entries of $\underTk e_k$\dots}

    ${\epsilon}_{k+1} = s_{k-1,1} \beta_{k+1}$
    \tcc*[r]{produces first two entries in $\underTkp e_{k+1}$}

    $\delta_{k+1} = -c_{k-1,1} \beta_{k+1}$

    $c_{k1}, s_{k1}, \gamma_k^{(2)}
      \leftarrow \SymOrtho(\gamma_k, \beta_{k+1})$
    \tcc*[r]{Current left reflection}

    $c_{k2}, s_{k2}, \gamma_{k-2}^{(6)}
      \leftarrow \SymOrtho(\gamma_{k-2}^{(5)}, \epsilon_k)$
    \tcc*[r]{First right reflection}

    $\delta_k^{(3)} = s_{k2}\vartheta_{k-1} - c_{k2} \delta_k^{(2)}$,
    \qquad $\gamma_k^{(3)} = -c_{k2}\gamma_k^{(2)}$,
    \qquad $\eta_k = s_{k2} \gamma_k^{(2)}$

    $\vartheta_{k-1}^{(2)} = c_{k2} \vartheta_{k-1}
                           + s_{k2} \delta_k^{(2)}$

    $c_{k3}, s_{k3}, \gamma_{k-1}^{(5)}
      \leftarrow \SymOrtho(\gamma_{k-1}^{(4)}, \delta_k^{(3)})$
    \tcc*[r]{Second right reflection\dots}

    $\vartheta_k = s_{k3}\gamma_k^{(3)}$,
    \qquad $\gamma_k^{(4)} = -c_{k3}\gamma_k^{(3)}$
    \tcc*[r]{to zero out $\delta_k^{(3)}$}

    $\tau_k = c_{k1} \phi_{k-1}$
    \tcc*[r]{Last element of $t_k$}

    $\phi_k = s_{k1} \phi_{k-1}$, \quad $\psi_{k-1} = \phi_{k-1}
         \norm{[\smash{\gamma_k \ \; \delta_{k+1}}]}$
    \tcc*[r]{Update $\norm{\rbar_k}$, $\norm{\tilde{A} \rbar_{k-1}}$}

    \lIf{$k=1$}{$\gamma_{\min}=\gamma_{1}$}
    \lElse{$\gamma_{\min} \leftarrow \min{ \{ \gamma_{\min},
           \gamma_{k-2}^{(6)}, \gamma_{k-1}^{(5)},
         | \gamma_k^{(4)} | \}}$}

    $\mathcal{A}_k = \max{ \{\mathcal{A}_{k-1}, \rho_k,
           \gamma_{k-2}^{(6)}, \gamma_{k-1}^{(5)}, |\gamma_k^{(4)}| \}}$
    \tcc*[r]{Update $\norm{\tilde{A}}$}

    $\omega_k = \norm{[\smash{\omega_{k-1} \ \; \tau_k}]}$,
    \qquad $\kappa_k \leftarrow \mathcal{A}_k / \gamma_{\min}$
    \tcc*[r]{Update $\norm{\tilde{A} x_k}$, $\cond(\tilde{A})$}

    $w_k = -(c_{k2} / \beta_k) q_k + s_{k2} w_{k-2}^{(3)}$
    \tcc*[f]{Update $w_{k-2}$, $w_{k-1}$, $w_k$}

    $w_{k-2}^{(4)} = (s_{k2} / \beta_k) q_k + c_{k2} w_{k-2}^{(3)}$

    \lIf{$k>2$}
        {$w_k^{(2)} = s_{k3} w_{k-1}^{(2)} - c_{k3} w_k$,
         \qquad $w_{k-1}^{(3)} = c_{k3} w_{k-1}^{(2)} + s_{k3} w_k$}

    \lIf{$k>2$}
        {$\mu_{k-2}^{(3)} = (\tau_{k-2} - \eta_{k-2} \mu_{k-4}^{(4)}
                                        - \vartheta_{k-2} \mu_{k-3}^{(3)})
                            / \gamma_{k-2}^{(6)}$}
    \tcc*[r]{Update $\mu_{k-2}$}

    \lIf{$k>1$}
        {$\mu_{k-1}^{(2)} =(\tau_{k-1} - \eta_{k-1} \mu_{k-3}^{(3)} -
         \vartheta_{k-1}^{(2)} \mu_{k-2}^{(3)}) / \gamma_{k-1}^{(5)}$}
    \tcc*[r]{Update $\mu_{k-1}$}

    \lIf{$\gamma_k^{(4)} \neq 0$}
        {$\mu_k = (\tau_k - \eta_k \mu_{k-2}^{(3)}
         - \vartheta_k \mu_{k-1}^{(2)}) / \gamma_k^{(4)}$}
    \lElse{$\mu_k = 0$}
    \tcc*[r]{Compute $\mu_k$}

    $x_{k-2}^{(2)}  = x_{k-3}^{(2)}  + \mu_{k-2}^{(3)} w_{k-2}^{(3)}$
    \tcc*[r]{Update $x_{k-2}$}

    $x_k = x_{k-2}^{(2)}  + \mu_{k-1}^{(2)} w_{k-1}^{(3)}
                       + \mu_k w_k^{(2)}$
    \tcc*[r]{Compute $x_{k}$}

    $\chi_{k-2}^{(2)} = \norm{[\smash{\chi_{k-3}^{(2)}  \
                               \; \mu_{k-2}^{(3)}}]}$
    \tcc*[r]{Update $\norm{x_{k-2}}$}

    $\chi_k = \norm{[\smash{\chi_{k-2}^{(2)}  \ \; \mu_{k-1}^{(2)}
                                       \ \; \mu_k}]}$
    \tcc*[r]{Compute $\norm{x_k}$}
  }

  \BlankLine

  $x = x_k$,
  \quad $\phi = \phi_k$,
  \quad $\psi = \phi_k \norm{[\smash{\gamma_{k+1}
                                      \ \; \delta_{k+2}}]}$,
  \quad $\chi = \chi_k$,
  \quad $\mathcal{A} = \mathcal{A}_k$,
  \quad $\omega = \omega_k$

  \Outputs{$x, \phi ,\psi, \chi, \mathcal{A}, \kappa, \omega$}
  \caption{Pseudocode of preconditioned MINRES-QLP for
    solving \newline $(A-\sigma I) x \approx b$.  In the
    right-justified comments, $\tilde{A} \equiv M^{-\myhalf} (A-\sigma
    I) M^{-\myhalf}$.}
  \label{algo-pminresqlp}
\end{algo}

\begin{algo}{tb}

  \Inputs{$a,b$}

  \smallskip

  \lIf{$b=0$}
      {$s=0$, \qquad $r=\left\vert a\right\vert$\;
      \qquad
      \lIf{$a=0$}
          {$c=1$}
      \lElse{$c=\operatorname{sign}(a)$}\;
      }
  \ElseIf{$a=0$}
         {$c=0$, \qquad $s=\operatorname{sign}(b)$, \qquad
         $r=\left\vert b\right\vert$}
  \ElseIf{$\left\vert b \right\vert \ge \left\vert a \right\vert$}
         {$\tau=a/b$, \qquad $s=\operatorname{sign}(b)/\sqrt{1+\tau^{2}}$,
         \qquad $c=s\tau$, \qquad $r=b/s$}
  \ElseIf{$\left\vert a \right\vert > \left\vert b \right\vert$}
         {$\tau=b/a$, \qquad $c=\operatorname{sign}(a)/\sqrt{1+\tau^{2}}$,
         \qquad $s=c\tau$, \qquad $r=a/c$}

  \Outputs{$c,s,r$}
  \caption{Algorithm SymOrtho.}
  \label{algo-symortho}
\end{algo}


\section{KEY FORTRAN~90 DESIGN FEATURES}

\begin{figure}[ht]   
%
\centering
\includegraphics[scale=.6]{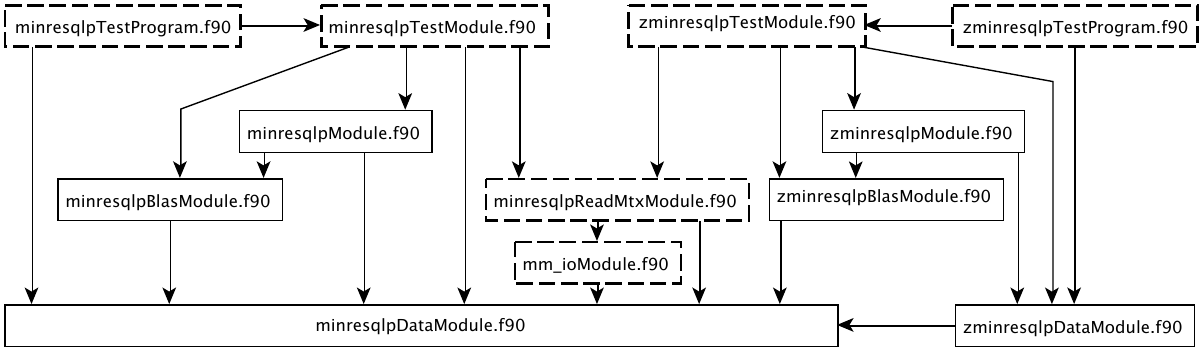}
\caption{FORTRAN~90 source files and their dependencies. Filenames
  boxed in broken lines are optional, and the corresponding files are
  used mainly for testing and demonstration.}
\label{fig:f90-program-units}
\end{figure}

Our \FORTRAN~90 package contains the following files for symmetric
problems with the first three files forming the core. Their
dependencies are depicted in Figure~\ref{fig:f90-program-units}.
\begin{enumerate}[\upshape 1.]
\item \texttt{minresqlpDataModule.f90}: defines precision and
  constants used in other modules
\item \texttt{minresqlpBlasModule.f90}: packages some BLAS
  functions~\citeA{netlibblas-f90}
\item \texttt{minresqlpModule.f90}: implements \MINRESQLP with
  preconditioning
\item \texttt{mm\_ioModule.f90} and
  \texttt{minresqlpReadMtxModule.f90}: packages subroutines for
  reading Matrix Market files~\citeA{MM,netlibblas-f90}
\item \texttt{minresqlpTestModule.f90}: illustrates how \MINRESQLP can
  call \texttt{Aprod} or \texttt{Msolve} with a short fixed parameter
  list, even if it needs arbitrary other data
\item \texttt{minresqlpTestProgram.f90}: contains the main driver
  program for unit tests
\item \texttt{Makefile}: compiles the \FORTRAN source files via the
  Unix command \texttt{make}
\item \texttt{minresqlp\_f90.README}: contains information about
  software license, other files in the package, and program
  compilation and execution.
\end{enumerate}
The counterparts of these programs for Hermitian problems have the
same filenames prefixed with the letter ``\texttt{z}''.

In our \FORTRAN~90 implementation, we use \emph{\texttt{modules}}
instead of the obsolete \FORTRAN~77 \texttt{COMMON} blocks for
grouping programs units and data together and controlling their
availability to other program units. A module can use \texttt{public}
data and subroutines from other modules (by declaring an
\emph{\texttt{interface}} block), share its own \texttt{public} data
and subroutines with other program units, and hide its own
\texttt{private} data and subroutines from being used by other program
units.

In \texttt{minresqlpModule.f90} we define a \texttt{public} subroutine
\texttt{MINRESQLP} that implements Algorithm~\ref{algo-pminresqlp}.
Two input arguments of this subroutine, \texttt{Aprod} and
\texttt{Msolve}, are external user-defined subroutines---we recommend
they be \texttt{private} for data integrity.  The subroutine
\texttt{Aprod} defines the matrix $A$ as an operator. For a given
vector $x$, the \FORTRAN statement \texttt{call Aprod(n, x, y)} must
return the product $y=Ax$ without altering $x$. The subroutine
\texttt{Msolve} is optional, and it defines a symmetric
positive-definite matrix as an operator $M$ that serves as a
preconditioner.  For a given vector $y$, the \FORTRAN statement
\texttt{call MSolve(n, y, x)} must solve the linear system $Mx=y$
without altering $y$.  To provide the compiler the necessary
information about these \texttt{private} subroutines defined in
\texttt{minresqlpTestModule}, an \texttt{interface} block in
subroutine \texttt{MINRESQLP} is declared, which essentially
replicates the headers of \texttt{Aprod} and \texttt{Msolve} in
\texttt{minresqlpTestModule}.

A public routine \texttt{minresqlptest}, also defined in module
\texttt{minresqlpTestModule}, calls \texttt{MINRESQLP} with
\texttt{Aprod} and \texttt{Msolve} passed to \texttt{MINRESQLP} as
parameters.

We declare all data variables in \texttt{minresqlpTestModule} used for
defining \texttt{Aprod} and \texttt{Msolve} to be \texttt{private} so
that they are accessible to all the subroutines in the module but not
outside.

To summarize, we have described and provided a pattern that allows
\MINRESQLP users to solve different problems by simply editing
\texttt{minresTestModule} (and possibly the main program
\texttt{minresTestProgram}, which calls \texttt{minresqlptest}).
Users do not need to change \texttt{MINRESQLP} as long as the header
of subroutines \texttt{Aprod} and \texttt{Msolve} stay the same in
\texttt{minresTestModule}.

Our design spares users from implementing \textit{reverse
  communication}, and hence enables the development of iterative
methods without \emph{a priori} knowledge of users' problem data $A$
and $M$ (by returning control to the calling program every time
\texttt{Aprod} or \texttt{Msolve} is to be invoked). While reverse
communication is widely used in scientific computing with \FORTRAN~77,
the resulting code usually appears formidable and unrecognizable from
the original pseudocode; see \cite{DEK95} and \cite{OS} for two
examples of \CG and numerical integration coded in \FORTRAN~77 and 90,
respectively. Our \MINRESQLP implementation achieves the purpose of
reverse communication while preserving code readability and thus
maintainability. The \FORTRAN~90 module structure allows a user's $Ax$
products and $Mx=y$ solves to be implemented outside \MINRESQLP in the
same way that \MATLAB's function handles operate.

Finally, unit testing is an important software development strategy
that cannot be overemphasized, especially in the scientific computing
communities. Unit testing usually consists of multiple small and fast
but specific and illuminating test cases that check whether the code
behaves as designed.  Software development is incremental, and errors
(also known as bugs) are often found over time. Adding new
functionalities or fixing errors often breaks the code for some
earlier successful test cases. It is therefore critical to expand the
test cases and to ensure that all unit tests are executed with
expected results every time a program unit is updated.

In our development of \FORTRAN~90 \MINRESQLP, we have created a suite
of $52$ test cases including singular matrices representative of
real-world applications~\cite{FSJSU,DH11}.  The test program outputs
results to \texttt{MINRESQLP.txt}. If users need to modify subroutine
\texttt{MINRESQLP}, they can run these test cases and search for the
word ``\texttt{appear}'' in the output file to check whether all tests
are reported to be successful.  For more sophisticated unit testing
frameworks employed in large-scale scientific software development,
see~\cite{OTL08}.

Further details on interface and implementation, with additional
numerical examples and documentation, are given in~\cite{CS12}.

As a last note, careful choices of parameter values are critical in
the convergence behavior of iterative solvers. While the default
parameter values in \MINRESQLP work well in most tests, they may need
to be fine-tuned in some cases by trial and error, solving a series of
problems as in iterative regularization, or partial or full
reorthogonalization of the Lanczos vectors.


\begin{acks}
  We thank Christopher Paige for his contribution to the theory of
  \MINRESQLP~\cite{CPS11}.  We also thank Tim Hopkins and David
  Saunders for testing and running our \FORTRAN~90 package on the
  Intel ifort compiler and the NAG Fortran compiler, resulting in more
  robust code. We are grateful to Zhaojun Bai and both anonymous
  reviewers for their patience and constructive comments. The first
  author also thanks Jed Brown, Ian Foster, Todd Munson, Gail Pieper,
  and Stefan Wild for their feedback and support during the
  development of this work.  We express our gratitude to the SIAM 2012
  SIAG/LA Prize Committee for their favorable consideration of
  \MINRESQLP~\cite{CPS11}.
\end{acks}

\bibliographystyle{acmtrans}
\bibliography{refs3} 

\begin{received}
Received 03 Aug 2011;
accepted 15 Jan 2012;
revised \today.
\end{received}

\end{document}